\documentclass[aps,nofootinbib,twocolumn,longbibliography,prd,showpacs]{revtex4-1}
\usepackage{amstext,amsmath,amssymb,amsfonts,bbm}
\usepackage[latin1]{inputenc}
\usepackage{graphicx}
\usepackage{tocvsec2}
\usepackage{color}

\usepackage{multirow} 
\usepackage{rotating}
\usepackage[colorlinks,citecolor=blue,linkcolor=blue,urlcolor=blue]{hyperref}
\usepackage{url}
\usepackage{mathrsfs}
\usepackage{fancyhdr}
\usepackage{hyperref}

\usepackage{lmodern}
\usepackage{textcomp}
\usepackage{dsfont}
\usepackage{natbib}

\usepackage{geometry}
\DeclareMathAlphabet{\mathpzc}{OT1}{pzc}{m}{it}

\newtheorem{theorem}{Theorem}[section]

\newtheorem{proposition}[theorem]{Proposition}

\newtheorem{postulate}[theorem]{Postulate}

\newenvironment{proof}[1][Proof.]{\begin{trivlist}
\item[\hskip \labelsep {\bfseries #1}]}{\end{trivlist}}
\newenvironment{definition}[1][Definition.]{\begin{trivlist}
\item[\hskip \labelsep {\bfseries #1}]}{\end{trivlist}}
\newenvironment{example}[1][Example.]{\begin{trivlist}
\item[\hskip \labelsep {\bfseries #1}]}{\end{trivlist}}
\newenvironment{remark}[1][Remark.]{\begin{trivlist}
\item[\hskip \labelsep {\bfseries #1}]}{\end{trivlist}}
\newenvironment{fact}[1][Fact.]{\begin{trivlist}
\item[\hskip \labelsep {\bfseries #1}]}{\end{trivlist}}

\begin{document}

\begin{center}
\title{Decoherent histories of spin networks}
\author{David P.B. Schroeren}
\date{\today}
\affiliation{Balliol College, Oxford OX1 3UH, UK}
\begin{abstract}
The decoherent histories formalism, developed by Griffiths, Gell-Mann, and Hartle \citep{Hartle:1992as, GellMann:2006uj, GellMann:2011dt, Griffiths2003} is a general framework in which to formulate a timeless, `generalised' quantum theory and extract predictions from it. Recent advances in spin foam models allow for loop gravity to be cast in this framework. In this paper, I propose a decoherence functional for loop gravity and interpret existing results \citep{Bianchi:2010zs, Bianchi:2011ym} as showing that coarse grained histories follow quasiclassical trajectories in the appropriate limit.\footnote{This paper was published in \textit{Foundations of Physics}; the final publication is available at \href{http://link.springer.com/article/10.1007/s10701-013-9698-4}{http://link.springer.com/}. Please cite the published version only.}
\end{abstract}

\maketitle

\end{center}

It has recently been shown that the space of boundary spin networks of spin foams is given by the kinematical Hilbert space of canonical loop gravity, in the context of the so-called EPRL-model \cite{Engle:2007wy}, and independently via operator spin foams \cite{Kaminski:2009fm}. Furthermore, there is good reason to believe that spin foams reproduce dynamics of canonical loop gravity, although at present there is no proof of this. Nonetheless, these developments allow for a tentative novel perspective on the problem of decoherence in quantum gravity.

In this paper, I focus on Hartle's definition of the class operator via path integrals \cite{Hartle:1992as}. Other proposals, most notably by J. Halliwell \emph{et al.} \citep{Halliwell2009, Halliwell2009a, Halliwell2010} involving a complex potential, as well as possible implications of the Zeno-effect for the class operators presented here will be the subject of an upcoming paper \cite{Schroeren2013}.

The implications of the decoherent histories framework for the foundations of physics have been widely discussed, for instance with respect to quantum measurement problem \citep{Wallace2010, Hartle2010, Halliwell2010a, Griffiths1995, Griffiths1984}, the arrival time problem \citep{Halliwell1998, Yearsley2011, Wallden2008}, as well as the problem of time in quantum gravity \citep{Anderson2012, Wallden2008, Gambini2005}. While the results of this paper are mostly formal, their foundational implications should be the subject of future work.

This paper is organised as follows: First, I give a brief review of the kinematics, dynamics, and the coheren states formalism of loop gravity. In section \ref{properties}, I review the decoherent histories formalism in its general form, before applying it to loop gravity in section \ref{spinnetworks}. In section \ref{coarsegraining} I discuss two types of coarse graining in loop gravity and go on to show that, under certain approximations, coarse grained histories follow quasi-classical trajectories in section \ref{qchistories}.

\section{The theory}
Here, I define the relevant aspects of loop quantum gravity (LQG). For comprehensive derivations of these expressions, see, for instance, \citep{Rovelli:2011eq, Dona:2010hm, thiemann2007modern} and references therein.
\subsection{Kinematical state space}
\begin{definition}
The diffeomorphism-invariant kinematical Hilbert space $\mathcal{H}_{kin}$ of LQG is defined by 
\begin{equation}
	\mathcal{H}_{kin} = \oplus_{\Gamma} \mathcal{H}_{\Gamma}
\label{eq:diffvhilb}
\end{equation}
where the sum runs over abstract graphs $\Gamma$ composed of a set of $L$ links $l$ and a set of $N$ nodes $n$, and the gauge-invariant Hilbert space $\mathcal{H}_{\Gamma}$ is defined by
\begin{equation}
	\mathcal{H}_{\Gamma} = L_2[SU(2)^L/SU(2)^N].
\label{eq:ginvhilb}
\end{equation}
which is the space of square-integrable functions on $SU(2)^L$ that are invariant under the transformations
\begin{equation}
	\psi(h_l) \rightarrow \psi(g_{s(l)}h_lg_{t(l)}^{-1}),
\label{eq:gaugetransform}
\end{equation}
where $s(l)$ and $t(l)$ denote the \emph{source} and \emph{target} of $l$, respectively, $g_{n} \in SU(2)^N$ and $h_l$ is the \emph{colouring} of the link $l$.
\end{definition}
\begin{definition}
The \emph{spin network basis} is a basis of the space $\mathcal{H}_{kin}$, elements of which are labelled by a graph $\Gamma$ and colourings $\sigma$ of links $l$ by $SU(2)$-irreps $j_l$ and nodes $n$ by intertwiners $\text{v}_n$;\footnote{See \cite{Rovelli:2011eq} for a description of this formalism.}  explicitly,
\begin{equation}
	\psi_{\Gamma, j_l, \text{v}_n}(h_l) = \left ( \otimes_l d_{j_l} D^{j_l}(h_l) \right ) \cdotp \left ( \otimes_n \text{v}_n \right ),
\label{eq:snstates}
\end{equation} 
\end{definition}
The diffeomorphism-invariant Hilbert space $\mathcal{H}_{kin}$ is the \emph{quantum configuration space} of loop gravity. The states $\psi_{\Gamma, j_l, \text{v}_n}(h_l) \in H_{kin}$ are the quantum states of a three-geometry or \emph{boundary states} of a four-region. Thus, $\mathcal{H}_{kin}$ corresponds to $\mathcal{H}^*_{\text{f}} \otimes \mathcal{H}_{\text{i}}$ of the initial and final states in nonrelativistic quantum theory. The $(\Gamma, j_l, \text{v}_n)$ are the `quantum numbers' of spin network states.

\subsection{Dynamics}
The spin foam programme is an attempt to obtain dynamics for spin networks via a path integral type formalism. Here I briefly recount the basic aspects of the construction, following \cite{Rovelli:2011eq} \cite{Perez2009}.

\begin{definition}
A \emph{two-complex} $\mathcal{C}$ is a combinatorial object composed of faces, edges and vertices. The boundary $\partial \mathcal{C}$ is a (possibly disconnected) graph $\Gamma$, whose links are edges of $\mathcal{C}$ bounding a single face and whose nodes are vertices of $\mathcal{C}$ bounding a single internal edge.
\end{definition} 

\begin{definition} Let $\mathcal{C}$ denote two-complex with $\Gamma = \partial \mathcal{C}$. For a group $G$, a \emph{spin foam} is a pairing $(\mathcal{C}, \sigma)$ of the two-complex $\mathcal{C}$ with a colouring $\sigma$, i.e. a labelling of each face $f$ with an irreducible representation $\rho_f$ of $G$ and a labelling of each edge $e \notin \Gamma$  with an intertwiner $\iota_e$.
\end{definition}
\begin{definition} The transition amplitude associated to a boundary spin network state $\psi$ with quantum numbers $(\Gamma, j_l$, $\text{v}_n)$ for 4d Lorentzian LQG where $G = SL(2,\mathbb{C})$ is given by:
\begin{align}
	\langle W | \psi \rangle &=& W(\Gamma, j_l, \text{v}_n) = \lim_{\mathcal{C} \rightarrow \infty, \partial \mathcal{C}= \Gamma} W_{\mathcal{C}}(j_l, \text{v}_n) \nonumber \\ 
	&=& \sum_{\mathcal{C}, \partial \mathcal{C}= \Gamma} W_{\mathcal{C}}(j_l, \text{v}_n),\qquad \qquad \qquad \qquad
\label{eq:fullphysicalampl}
\end{align}
where the last equality follows from the fact that, due to diffeomorphism invariance, summing is refining if cylindrical consistency holds (cf. Smerlak and Rovelli in \cite{Rovelli:2010qx}). The transition amplitude truncated to a two-complex $\mathcal{C}$ with boundary $\partial \mathcal{C}$ in expression (\ref{eq:fullphysicalampl}) are given by
\begin{align}
	 W_{\mathcal{C}}(j_l, \text{v}_n) =  \sum_{j_f,\text{v}_e} \prod_f (2j_f + 1) \prod_v A_v(j_f, \text{v}_n), 
	 \label{eq:transitionampl} 
\end{align}
where $\mathcal{C}$ is a two-complex bounded by the graph $\partial \mathcal{C}$, with faces $f$, edges $e$ and vertices $v$. The vertex amplitude $A_v(j_{f}, \text{v}_e)$ is given by
\begin{equation}
	A_v(j_{f}, \text{v}_e) = tr[\otimes_{e\in v}(f_{\gamma} \text{v}_n)],
\label{eq:vertexampl}
\end{equation}
where $f_{\gamma}$ is defined in \cite{Rovelli:2011eq}. 
\end{definition}
The quantities $\langle W | \psi \rangle$ are complex numbers that will generally depend on properties of the boundary spin network $\psi$, for instance its graph and colouring. Equation (\ref{eq:fullphysicalampl}) can be understood as \emph{projecting onto the physical states} among the kinematical states, or alternatively as \emph{projecting out the remaining gauge symmetry}, the resulting states of which are the physical states (which have yet to be shown to lie in the kernel of the Hamiltonian constraint - this is hoped for but as yet unproven). 
\begin{example}
It will be instructive to consider the case where the boundary spin network is disconnected into two connected components on graphs $\Gamma, \Gamma'$, respectively, which I denote by $|\Gamma, j_l, \text{v}_n \rangle$ and $| \Gamma', j'_l, \text{v}'_n \rangle$. Then, the transition amplitude for this process can be written as
\begin{align}
		W(\Gamma \cup \Gamma', j_l, j'_l, \text{v}_n, \text{v}'_n) =	  \langle \Gamma, j_l, \text{v}_n  | \Gamma', j'_l, \text{v}'_n \rangle_{phys}
\label{eq:fghistspinnet}
\end{align}
where $W(\Gamma \cup \Gamma', j_l, j'_l, \text{v}_n, \text{v}'_n)$ denotes the spin foam sum over all two-complices that have $\Gamma \cup \Gamma'$ as their boundary, and $\langle \cdot | \cdot \rangle_{phys}$ denotes the physical inner product. One can think of the primed state as the `final' state and of the unprimed state as the `initial state'. However, in the fundamental theory there is no relation such as `earlier than' and `later than', and in general, the boundary state may be disconnected in more than two connected components, or not disconnected at all. Hence, this terminology should be taken with caution.
\end{example}
\subsection{Physical picture}
Let me illustrate the above using a language similar to transition amplitudes in ordinary quantum theory, following \cite[pp.~320]{rovelli2004quantum} and \cite{Oriti:2001qu}.

Suppose that a system's dynamical variable is denoted by $x$. The transition between two states $(q_f,T)$ and $(q_i,0)$ is given by the so-called Feynman path integral
\begin{equation}
\label{eq:feynmanpathintegral}
	\langle q_\text{f}, T | q_i, 0 \rangle = \int_{q_\text{f}, q_i} \delta q e^{i S[q(\tau)]},
\end{equation}
a sum over paths $q(\tau)$ that start at $q_i$ and end at $q_f$. $S[q(\tau)]$ is the action of this path. This can be defined in the Hamiltonian (canonical) theory by expressing the time evolution operator $e^{-iHt}$ as a product of operators for $N$ small time intervals, before taking the limit $dt = T/N \rightarrow 0$ or $N\rightarrow 0$.
\begin{align}
	\int_{q_\text{f}, q_\text{i}} \delta q e^{i S[q(\tau)]} = \lim_{N \rightarrow \infty} \int dq_1...dq_{N-1} \nonumber \\
	\langle q_\text{f} | e^{-iH \frac{T}{N}}|q_{N-1} \rangle \langle q_{N-1} | e^{-iH \frac{T}{N}} | q_{N-2} \rangle \nonumber \\ ... \langle q_1 | e^{-iH \frac{T}{N}} | q_\text{i} \rangle
\label{eq:pathintegraldecompose}
\end{align}

The formal analogy of this equation in LQG is (\ref{eq:fullphysicalampl}). This can be motivated as follows: Suppose that the Hamiltonian constraint operator $H = \int dx H(x)$ (where $x$ is a label for configuration space variables) has a non-negative spectrum, and consider the transition between two spin network states $\psi, \psi'$. The ``evolution operator'' is given by 
\begin{equation}
	P = \lim_{t\rightarrow \infty} e^{-i\int dx H(x) t}.
	\label{eq:evop}
\end{equation}
Thus, the full physical transition (\ref{eq:fullphysicalampl}) between $\psi, \psi'$ can be written as
\begin{equation}
	\langle \psi | \psi' \rangle_{phys} = \lim_{t\rightarrow \infty} \langle \psi | e^{-i\int dx H(x) t} | \psi' \rangle_{kin}
\label{eq:evopampl}
\end{equation}
Since $H$ diff-invariantly generates 4d-propagation, we can drop the infinite limit:(cf. \cite[p.324]{rovelli2004quantum})
\begin{equation}
	\langle \psi | \psi' \rangle_{phys} = \langle \psi | e^{-i\int_0^1 dt\int dx H(x) t} | \psi' \rangle_{kin}
\label{eq:evopampl2}
\end{equation}
Expanding this, we get
\begin{align}
	\langle \psi | \psi' \rangle_{phys} =  \lim_{N \rightarrow \infty} \sum_{\phi_1...\phi_{N-1}} \langle \psi | e^{-i\int dx H(x) dt} | \phi_{N-1} \rangle \nonumber \\
	 ... \langle \phi_1 | e^{-\int dx H(x) dt}| \psi' \rangle, \qquad \qquad \qquad
\label{eq:evopampl3}
\end{align}
wherein all $\langle \cdot |\cdot \rangle$ denote kinematical inner products. Now, expanding the exponential around small $dt$ for fixed $N$ yields terms equivalent to histories with lower $N$. Hence, in loop gravity, the infinite limit of a sum over spin network histories becomes simply a discrete sum over spin network histories \emph{with arbitrary length}. The contribution of each ``step'' along the history is given by a vertex amplitude (\ref{eq:vertexampl}), that is 
\begin{equation}
	A_{v}(\phi_{N+1},\phi_{N}) \sim \langle \phi_{N+1} | e^{-i\int dx H(x) dt} | \phi_{N} \rangle.
\label{eq:vertexamplevop}
\end{equation}

Hence, spin foams are sums over discrete spin network histories $(\phi_N,...,\phi_1)$ for all possible configurations and two-complices. The histories are generated by individual, discrete steps, since the Hamiltonian operator acts only on the nodes of a spin network (roughly, by creating an arc in its vicinity, cf. \cite{thiemann2007modern} for details). The amplitude of each step is given by (\ref{eq:vertexamplevop}), the matrix element of $H$; multiples of such vertex contributions (up to weight factors and face amplitudes) make up the spin foam amplitude (\ref{eq:fullphysicalampl}).

\subsection{Coherent states}
Generic states of a given geometry are given by linear superpositions of spin network states. Coherent states are introduced as a means to `dequantise' a quantum theory; that is, to construct states which are peaked around values of the conjugate variables of the corresponding classical system. The precise definition of coherent states for LQG can be found in \cite{thiemann2007modern}; here, it shall suffice to note that LQG coherent states are of the form 
\begin{equation}
	\psi_{H_l}(h_l) = \langle h_l | \psi_{H_l} \rangle, 
\label{eq:holcoh}
\end{equation}
as shown in \cite{Bianchi:2009ky}. The $\psi_{H_l}(h_l)$ are holomorphic functions of $H_l \in SL(2,\mathbb{C})$. These are Thiemann's complexifier coherent states in the case where the complexifier is the $SU(2)$ Laplacian $\Delta_{SU(2)}$ with eigenvalues of the form $j(j+1)$. Explicitly,
\begin{equation}
	\psi_{H_l}(h_l) = \int_{SU(2)}{dg_n \otimes_{l \in \Gamma} K_t(g_{s(l)}H_l g^{-1} h_l^{-1}}),
\label{eq:coherent}
\end{equation} 
are coherent spin-network states, where $K_t$ denotes the analytic continuation to $SL(2,\mathbb{C})$ of the $SU(2)$ heat kernel, given by
\begin{equation}
	K_t(g) = \sum_j (2j+1) e^{-j(j+1)t} Tr[D^j(g)],
\label{eq:heatkernel}
\end{equation} where the parameter $t$ is the \emph{heat kernel time}. The states (\ref{eq:holcoh}) have the usual peakedness properties and have a small spread around extrinsic and intrinsic curvature.

The truncated spin foam with a boundary coherent state $\psi_{H_l}$ given by
\begin{equation}
	W_{\mathcal{C}}(H_l) = \sum_{H'_l} \prod_f (2j_f + 1) \prod_v A_v(H'_l),
\label{eq:cohampl}
\end{equation}
where (see \cite[p.~18]{Rovelli:2011eq})
\begin{equation}
	A_v(H_l) = \int_{SL(2,\mathbb{C})^N} d\tilde{g}_n \prod_l P(H_l, g_{s(l)}g_{t(l)}^{-1}),
\label{eq:vertexampl1}
\end{equation}
and, for $Y_{\gamma}$ as defined in \cite[p.~12]{Rovelli:2011eq},
\begin{align}
	P(H,g) =  \sum_j (2j+1) e^{- j(j+1)t}  \nonumber \\ 
	\times \text{ tr}[D^j(H) Y_{\gamma}^{\dagger} D^{(\gamma j, j)}(g) Y_{\gamma}], \label{eq:projector1}
\end{align}
where In the above, $D^j(H)$ denotes the analytic continuation to $SL(2,\mathbb{C})$ of the $SU(2)$ Wigner matrix.

Explicitly, we can decompose any $SL(2,\mathbb{C})$-label in equation (\ref{eq:coherent}) as follows: (cf. \cite[p.~17]{Rovelli:2011eq})
\begin{equation}
	H_l = e^{2itL_l}h_l,
\label{eq:sl2clabel}
\end{equation}
where $h_l\in SU(2)$ is the holonomy along an edge encoding the intrinsic curvature and $L_l \in \mathfrak{su}(2)$ encodes the extrinsic curvature. The operators $\hat{h}_l$ and $\hat{L}_l$ have expectation values given by
\begin{equation}
	\frac{\langle \psi_{H_l}|\hat{h}_l | \psi_{H_l} \rangle}{\langle \psi_{H_l}| \psi_{H_l} \rangle} = h_l, \text{ } \frac{\langle \psi_{H_l}|\hat{L}_l | \psi_{H_l} \rangle}{\langle \psi_{H_l}| \psi_{H_l} \rangle} = L_l,
\label{eq:expval}
\end{equation}
and the spread of the $|\psi_{H_l}\rangle $ around these values goes to zero with $\hbar$; that is, the $|\psi_{H_l}\rangle $ are coherent states peaked on values of $h_l, L_l$.

The full physical transition amplitude for the coherent state $| \psi_{H_l}\rangle = | \Gamma, H_l \rangle $ labelled by $H_l \in SL(2,\mathcal{C})$ on the boundary graph $\Gamma$ is given by 
\begin{align}
	\langle W | \psi_{H_l} \rangle = W(\Gamma, H_l)& =& \lim_{\mathcal{C} \rightarrow \infty, \partial \mathcal{C}= \Gamma} W_{\mathcal{C}}(H_l) \nonumber \\ 
	&=& \sum_{\mathcal{C}, \partial \mathcal{C}= \Gamma} W_{\mathcal{C}}(H_l).
\label{eq:fullphysicalamplcoh}
\end{align}
\section{Generalised quantum mechanics}	\label{properties}
The decoherent (or consistent) histories formalism has been developed as a means of clarifying how quantum mechanics assigns probabilities to macroscopic, mutually incompatible events, e.g. measurement outcomes. More precisely, according to the decoherent histories formalism, the predictions of quantum theory are \emph{in fact} assignments of probabilities to alternative coarse grained histories of a closed system, for an exhaustive set of such histories. Consistency with probability calculus is ensured by assigning probabilities only to \emph{decoherent} or \emph{consistent} sets of histories, i.e. sets for whose members interference can be neglected. 

The aim of this programme, then, is twofold: first, to reformulate quantum theory so as to be applicable to closed systems such as the universe, eliminating an unphysical distinction between the `observer' and the `observed' of Copenhagen quantum theory; secondly, to address the `problem of time' in quantum gravity. 

In this section, I introduce what Hartle and Gell-Mann refer to as \emph{generalised quantum mechanics}, essentially following their \cite{GellMann:2011dt} and Hartle's \cite{Hartle:1992as}. One way to present this is to describe the set of alternatives at any given moment of time by a set of orthogonal Heisenberg projectors $\{ P^k_{\alpha_k}(t_k)\}$ which, for a sequence of times $(t_1,...,t_n)$ define a set of alternative histories for the system, specified by chains of alternatives $(\alpha_1,...,\alpha_n)$. However, the notion of a `sequence' of times is not well defined in a quantum theory of gravity as there is no covariant notion of an alternative at any instance of time, so I move directly to the Feynman path integral version of the formalism, anticipating its adaptation to loop gravity in the final section.

\subsection{Decoherent histories framework}
\begin{definition} Let $\mathcal{H}$ denote a Hilbert space describing a quantum system. The set $\mathfrak{f}$ of \emph{fine grained histories} $f$ is specified by giving a maximally detailed description of the system in terms of its state space $\mathcal{H}$ and its dynamics. In standard quantum theory, this is given by the unitary evolution of rays in a Hilbert space according to the Schrödinger equation; in a quantum field theory, by histories of field configurations.
\end{definition}

\begin{definition}
Let $\mathcal{H}$, $\mathfrak{f}$ be as before. A \emph{coarse graining} on the fine grained history space $\mathfrak{f}$ is defined as any partition of $\mathfrak{f}$ into an exhaustive set $\{ c_{\alpha} \} $ of exclusive classes $c_{\alpha}$ which obey
\begin{equation}
	\cup_{\alpha}c_{\alpha} = \mathfrak{f}
\label{eq:coarsegr}
\end{equation}
Each class $c_{\alpha}$ is a \emph{coarse grained history}, and the set of classes $\{ c_{\alpha} \} $ is a set of coarse grained histories, which I will also refer to simply as $\{\alpha \}$. Further partitioning can be imposed on the set of coarse grained histories; a process which terminates with a set that contains only a single member (the trivial coarse graining). A \emph{fine graining} is defined analogously.
\end{definition}
\begin{fact}
The operations of coarse graining and fine graining define a partial ordering on the set of all sets of histories.
\end{fact}

Associated with a coarse graining is the transition amplitudes of coarse grained histories, defined as follows:
\begin{definition}
The \emph{class operator} $C_{\alpha}$ corresponding to a coarse grained history $c_{\alpha}$ is defined as
\begin{equation}
	\langle \psi_{\text{f}} | C_{\alpha} | \psi_{\text{i}} \rangle := \int_{\psi_{\text{f}}, \alpha, \psi_{\text{i}}} \mathcal{D}\phi e^{i\mathcal{S}(\phi)},
\label{eq:classop}
\end{equation}
where $\mathcal{S}$ is the action, $\phi$ ranges over configuration space, $\psi_{\text{i}}, \psi_{\text{f}} \in \mathcal{H}$ are \emph{initial} and \emph{final} states, and $\mathcal{D}\phi$ denotes the functional integration measure. The state obtained by the action of the class operator on generic initial states $|\psi \rangle$ is called the \emph{branch state vector}  $|\psi_{\alpha} \rangle := C_{\alpha} |\psi \rangle$.
\end{definition}

It is important to note that the space of fine grained histories is deterministic and therefore \emph{trivially decoherent}. However, this is generally not true of coarse grained histories.

Probabilities are assigned only to \emph{decoherent} histories - that is, to histories whose interference is negligible. Interference is measured by the decoherence functional, which is defined as follows.\\

\begin{definition}
Let $\{\alpha \}$ denote a set of histories (possibly coarse grained) and let $\alpha, \alpha' \in \{\alpha \}$. Further, denote by $|\psi_{\alpha} \rangle = C_{\alpha} |\psi \rangle$ the \emph{branch state vector} of $\alpha$.  The \emph{decoherence functional} $D(\alpha, \alpha'):= \langle \psi_{\alpha} | \psi_{\alpha'} \rangle$ is a complex-valued functional $D: \{\alpha \} \times \{\alpha \} \rightarrow \mathbb{C}$ which obeys

\begin{enumerate}

	\item \emph{Hermiticity. } $D(\alpha', \alpha)  = D^*(\alpha, \alpha')$
	\item \emph{Positivity. } $D(\alpha', \alpha)  \geq 0$
	\item \emph{Normalisation. } $\sum_{\alpha, \alpha'} D(\alpha', \alpha)  = 1$
	\item \emph{Superposition principle. } $D(\bar{\alpha}', \bar{\alpha})  = \sum_{\alpha' \in \bar{\alpha}'} \sum_{\alpha \in \bar{\alpha}} D(\alpha', \alpha)$,

\end{enumerate}
for any history space $\{\bar{\alpha} \}$ which is a coarse graining of $\{\alpha \}$ and $\bar{\alpha}', \bar{\alpha} \in \{\bar{ \alpha} \}$.
\end{definition}
Importantly, condition (4) of the above entails that, given the decoherence functional for a fine grained history space, one can obtain the decoherence functional for any coarse graining via the superposition principle (cf. \cite[p.~52]{Hartle:1992as}).

\begin{postulate} \emph{(Medium decoherence condition)}\\
The subset of the set of all sets of alternative coarse grained histories to which probabilities are assigned is picked out by the \emph{medium decoherence condition} 
\begin{equation}
	D(\alpha, \alpha') \approx 0, \text{ }\forall \alpha \neq \alpha' \in \{ \alpha \}
\label{eq:medium}
\end{equation}
These histories are referred to as \emph{decoherent}.
\end{postulate}
Accurately speaking, histories that obey equation (\ref{eq:medium}) decohere \emph{approximately}, and thus the medium decoherence condition ensures that such histories conform to probability calculus only aproximately. However, as Hartle remarks, ``when we speak of approximate decoherence and approximate probabilities we mean decoherence achieved and probability sum rules satisfied beyond any standard that might be conceivably contemplated for the accuracy of prediction and the comparison of theory with experiment.'' (\cite[p.~20]{Hartle:1992as}.)

For a set of decoherent histories $\{ \alpha \}$, the probability for a particular history is given by the `diagonal' elements of the decoherence functional
\begin{equation}
	p(\alpha) = D(\alpha, \alpha)
\label{eq:prob}
\end{equation}

In addition, we can observe the following
\begin{fact}
Let $\{ \alpha \}$, $\{ \bar{\alpha} \}$ as above. Probabilities defined by equation (\ref{eq:prob}) must be additive on disjoint sets of the sample space, that is
\begin{equation}
	p(\bar{\alpha}) = \sum_{\alpha \in \bar{\alpha}} p(\alpha).
\label{eq:disjadd}
\end{equation}
Furthermore, the probabilities given by equation (\ref{eq:prob}) obey probability calculus.
\end{fact}
\subsection{Application to general relativity}
I follow section VIII.4 of \cite{Hartle:1992as} and consider the case without matter. In the Hamiltonian formulation of general relativity, the Einstein Hilbert action is written in a $3+1$ form in terms of lapse $N$, shift $N^i$, induced three-metric $q_{ij}$ - the configuration variable - and the conjugate momentum $\pi_{ij}$ determined by the extrinsic curvature of foliating three-surfaces (see chapter 1 of \cite{thiemann2007modern} for details). 

In this framework, fine grained histories are configurations of three-metrics between initial and final metrics $q_{ij}, q_{ij}'$ on the boundaries $S_f$ and $S_i$, respectively. For some coarse graining $\{c_{\alpha} \}$, Hartle defines the class operator as
\begin{equation}
	\langle q_{ij}' | C_{\alpha} | q_{ij} \rangle = \int_{\alpha} \delta q \delta \pi \delta N e^{iS[N^{\beta}, \pi_{ij}, q_{ij}]},
\label{eq:grclassop}
\end{equation}
where suitable measure factors have been supressed, $N^{\beta} = (N, N^i)$ and $S$ is the ADM action.

For a pure initial state $\psi$, the decoherence functional for sets $\{\psi_i \}, \{\phi_j\}$ of final and initial states, respectively, is then given as
\begin{equation}
	D(\alpha, \alpha') = \mathcal{N} \sum_{ij} p_i' p_j \langle \psi_j | C_{\alpha'} | \phi_i \rangle \langle \phi_i | C_{\alpha}^{\dagger} | \psi_j \rangle,
\label{eq:grdecfunc}
\end{equation}
where $p_i' p_j$ are the coefficients of the respective density matrices for initial and final states and $\mathcal{N}$ is a suitable normalization (cf. \cite[p.~143]{Hartle:1992as}).

As such, these expressions are purely formal, as the notion of a quantum state of a three-geometry is left unspecified, and the inner product between wave functionals $\phi := \phi(q_{ij})$ on superspace is not well-defined. 

These issues are resolved in loop gravity, where both the state of a three geometry as well as the inner product are well defined. 
\section{Decoherent histories of spin networks} \label{spinnetworks}
Given the formal presentation of the theory in section I, I now turn to the formulation of loop gravity as a generalised quantum mechanics described above. Most of this is a straightforward result of the covariant loop gravity dynamics of spin foams.
\subsection{Fine grained histories}
In loop gravity, fine-grained histories are given as histories of specific quantum states of three-geometries, rigorously defined by spin-network states. 

\begin{definition} \emph{(Fine grained histories.)} \\
Let $(\Gamma = \partial \mathcal{C}, j_l, \text{v}_n)$ be the quantum numbers specifying a spin network $\psi$, and denote the colouring of $\Gamma$ by $\sigma_{B} := (j_l, \text{v}_n)$. Its fine grained history is given by a spin foam $(\mathcal{C}, \sigma)$ given by a two complex $\mathcal{C}$ with boundary $\Gamma$, $N$ vertices $v$ and a colouring $\sigma$ of its faces, edges and vertices such that the colouring of the boundary graph is $\sigma_{B}$. The amplitude of the fine grained history $\psi$ is given by an $N$-product of vertex amplitudes $A_{v}$ as
\begin{equation}
	W_{\mathcal{C}}(\sigma_{B}, \sigma)= \prod_f (2j_f + 1) \prod_v^N A_v(j_f, \text{v}_n).
\label{eq:finegrained}
\end{equation}
Fine grained histories of spin networks can be thought of as `sweeping out' curves in the quantum configuration space $\mathcal{H}_{kin}$.
\end{definition}
\begin{remark}
Equation (\ref{eq:finegrained}) is identical to the truncated transition amplitude (\ref{eq:transitionampl}) on a two-complex $\mathcal{C}$, but without the sum over all possible configurations (i.e. colourings) of the two complex. That is, 
\begin{equation}
	\langle W | \psi(\sigma_{B}) \rangle = \sum_{\mathcal{C}} \sum_{\sigma}W_{\mathcal{C}}(\sigma_{B}, \sigma)
\label{eq:sumoverfgh}
\end{equation}
again gives the full physical transition amplitude for the boundary state $\psi(\sigma)$.
\end{remark}

\subsection{Coarse grained histories}
As Hartle points out, a coarse graining for the history space of a quantum theory of spacetime geometry fundamentally consists in specifying a set of assertions which partition the history space into classes of histories where the assertions are true and where they are false. That is, `every assertion about the universe [...] is the assertion that the history of the universe lies in a coarse-grained class in which the assertion is true and not in which it is false.' (\cite[p.~133]{Hartle:1992as}) Allowed coarse grainings thus amount to partitions of fine grained history space into an `exhaustive set of exclusive, \emph{diffeomorphism-invariant} classes.' (ibid.) In the case of covariant loop gravity, the following definition of coarse graining is natural:

\begin{definition} \emph{(Coarse graining.)}\\
Let $\mathfrak{f}$ denote the set of fine grained histories $f$ specified by physical transition amplitudes $W_{\mathcal{C}} ( \Gamma, j_l, \text{v}_n)$ for boundary spin networks $|\Gamma, j_l, \text{v}_n \rangle$ with $\partial \mathcal{C} \equiv \Gamma$. A \emph{coarse graining} of $\mathfrak{f}$ consists in specifying a list $\alpha$ of $n$ diff-invariant properties $\alpha := (\alpha_1,...,\alpha_n)$ such that the history space $\mathfrak{f}$ partitions into classes $c_{\alpha}$ for which every $f \in c_{\alpha}$ satisfies the properties encoded in $\alpha$.

\end{definition}
\begin{remark}
The above definition of coarse graining for loop gravity ensures that (1) the partition into classes $c_{\alpha}$ is diff-invariant by virtue of the properties being defined as diff-invariant and (2) that the classes $c_{\alpha}$ are exclusive, since the properties can be chosen such that there is no $f \in \mathfrak{f}$ such that $f$ possess the properties encoded in $\alpha$ and $f \notin c_{\alpha}$.
\end{remark}
\begin{remark}
A general coarse graining may, for instance, be imposed by considering all fine grained histories which `pass through' a region in the quantum configuration space $\mathcal{H}_{kin}$ at least once; that is, for which at least one spin network along its history lies in the region of the kinematical Hilbert space selected by the coarse graining.
\end{remark}

To motivate the following proposition of a class operator for LQG, consider the Hilbert space of a single particle $\mathcal{H} = L_2[\mathbb{R}]$ moving in one dimension. As discussed above, the full transition amplitude between two points $q_i, q_f$ is given by a sum over the contributions of all possible histories of the system
\begin{equation}
	\langle q_\text{f}, T | q_\text{i}, 0 \rangle = \int_{q_\text{f}, q_\text{i}} \delta q e^{i S[q(\tau)]},
\label{eq:oneparticleampl}
\end{equation}
where the states $| q,t \rangle$ are Heisenberg states and the functional integration ranges over all paths which connect the two points. Suppose that we coarse grain by partitioning the set of fine grained histories into those that pass through an interval $ \Delta \in \mathbb{R}$ of the real line at a fixed time $t$ and those that do not, the class operator is given by the path integral over all those paths which, at the time $t$, pass through the interval $\Delta$
\begin{equation}
	\langle q_\text{f}, T | C_{\alpha} | q_\text{i}, 0 \rangle = \int_{q_\text{f}, \alpha, q_\text{i}} \delta q e^{iS[q]},
\label{eq:oneparticleclassop}
\end{equation}
which can be written as the product of the Feynman path integrals from $q_i$ to $q'$ and from $q'$ to $q_f$, summed over all $q' \in \Delta$:
\begin{equation}
		\langle q_\text{f}, T | C_{\alpha} | q_\text{i}, 0 \rangle = \int_{\Delta} dq' \int_{q_\text{f}, q'} \delta q e^{iS[q]}\int_{q', q_\text{i}} \delta q e^{iS[q]}
\label{eq:oneparticleclassop2}
\end{equation}
That is, we sum over all those histories which satisfy a certain condition, e.g. of passing through an interval of the real line at a fixed time. This idea can be applied to loop gravity, noting that (1) the fine grained history of a spin network is given by a coloured two-complex with boundary, its amplitude by (\ref{eq:finegrained}), and due to diff-invariance, there is no time ordering on the properties encoded in $\alpha$\footnote{Notwithstanding Savvidou \cite{Savvidou:2004vg}.}; (2) due to diff-invariance, the notion of \emph{intermediate state} of a history does not make sense - rather, I will speak of \emph{bulk configurations} on which the conditions are imposed. Bulk configurations are specified by two-complices and face/edge colourings.
\begin{proposition} \emph{(Class operator.)}\\
Let $f \in \mathfrak{f}$ be as before, and suppose that the kinematical Hilbert space has the structure $\mathcal{H}_{kin} = \mathcal{H}_{\text{\emph{f}}}^* \otimes \mathcal{H}_{\text{\emph{i}}}$. Further suppose that $\psi_\text{\emph{f}} \otimes \psi_\text{\emph{i}} \in \mathcal{H}_{kin}$ denote the boundary state of a spin foam with transition amplitude $W(\psi_\text{\emph{f}}, \psi_\text{\emph{i}}) = \langle \psi_\text{\emph{f}} | \psi_\text{\emph{i}} \rangle_{phys}$. 

Let $\alpha = (\alpha_1,...,\alpha_n)$ denote a list of diff-invariant properties and $\{c_{\alpha}\}$ the associated exhaustive diff-invariant set of exclusive classes $c_{\alpha}$ of histories. The \emph{class operator} for this coarse graining is by matrix elements
\begin{eqnarray}
	\langle \psi_\text{\emph{f}} |C_{\alpha} | \psi_{\emph{i}} \rangle & = &\sum_{\mathcal{C}_{\alpha},\sigma_\alpha}W(\sigma_B, \sigma) \qquad \qquad \nonumber \\
	& = &\sum_{\mathcal{C}_{\alpha},\sigma_\alpha} \prod_f (2j_f + 1) \prod_v A_v(\sigma)
	\label{eq:coarsegraining}
\end{eqnarray}
where sum runs over all bulk configurations for which the properties $\alpha$ are satisfied.
\end{proposition}
\begin{remark} The restricted path integral in equation (\ref{eq:coarsegraining}) is the sum over all those paths which satisfy the conditions encoded in $\alpha$; equivalently, the sum over all those fine grained histories which `pass through' a region of $\mathcal{H}_{kin}$ which is selected by $\alpha$.
\end{remark}
\begin{proof} Equation (\ref{eq:coarsegraining}) is a sum over histories rectricted by the diff-invariant partition $\{c_{\alpha}\}$. Summing over all $\alpha$ 
\begin{align}
	\sum_{\alpha} \langle \psi_\text{f} |C_{\alpha} | \psi_\text{i} \rangle &=& \sum_{\alpha} \sum_{\mathcal{C}_{\alpha},\sigma_{\alpha}} \prod_f (2j_f + 1) \prod_v A_v(\sigma) \nonumber \\
  & = & \sum_{\mathcal{C},\sigma} \prod_f (2j_f + 1) \prod_v A_v(\sigma) \nonumber \\
	& = &  \langle \psi_\text{f} |  \psi_\text{i} \rangle_{phys} \qquad \qquad \qquad 
\label{eq:proof1}
\end{align}
gives the full physical transition amplitude between the initial state $\psi_\text{i}$ and the final state $\psi_\text{f}$. $\blacksquare$
\end{proof}
\subsection{Decoherence functional}
A decoherence functional is readily constructed from the class operator above by multiplying the restricted amplitudes (\ref{eq:coarsegraining}) and summing over all possible boundary states $\psi = \psi_\text{f} \otimes \psi_\text{i}$. 

\begin{proposition} \emph{(Decoherence functional.)}
Let $f \in \mathfrak{f}$ denote the fine grained history space, and suppose that the kinematical Hilbert space has the structure $\mathcal{H}_{kin} = \mathcal{H}_{\text{\emph{f}}}^* \otimes \mathcal{H}_{\text{\emph{i}}}$. Further suppose that $\psi_\text{\emph{f}} \otimes \psi_\text{\emph{i}} \in \mathcal{H}_{kin}$ denote the boundary states of a spin foam with transition amplitude $W(\psi_\text{\emph{f}}, \psi_\text{\emph{i}}) = \langle \psi_\text{\emph{f}} | \psi_\text{\emph{i}} \rangle_{phys}$. Let $\{c_{\alpha}\}$ denote an exhaustive diff-invariant set of exclusive classes $c_{\alpha}$ of histories $f \in \mathfrak{f}$ specified by lists $ \alpha  = (\alpha_1,...,\alpha_n)$ of diff-invariant properties. Then, the decoherence functional for coarse grained histories of spin networks is given by
\begin{align}
	D(\alpha, \alpha') = \mathcal{N} \sum_{\psi_\text{\emph{f}} \otimes \psi_\text{\emph{i}}} \langle \psi_\text{\emph{f}} |C_{\alpha} | \psi_\text{\emph{i}} \rangle \langle \psi_\text{\emph{i}} |C_{\alpha'} | \psi_\text{\emph{f}} \rangle, 
	\label{eq:generaldecfunc}
\end{align}
for a suitable normalization $\mathcal{N}$.
\end{proposition}
\begin{proof}
The expression (\ref{eq:generaldecfunc}) is manifestly hermitian; positivity follows from the positivity of the physical inner product $\langle \cdot | \cdot \rangle_{phys}$. 

To show that the superposition condition holds, suppose that (as before), $\mathfrak{f}$ denotes the space of fine grained spin network histories and $f, f'\in \mathfrak{f}$. For any $\alpha \in \{c_{\alpha}\}$, there will generally be a multitude of fine grained histories $f$ compatible with $\alpha$ such that $f \in \alpha$. Furthermore, as noted above, coherent spin network states are generally superpositions of spin network basis states. Thus, it follows that 
\begin{equation}
	D(\alpha, \alpha') = \sum_{f \in \alpha} \sum_{f' \in \alpha'} D(f, f').
\label{eq:superpositionfc}
\end{equation}
Finally, if the RHS of the expression 
\begin{equation}
D(\alpha, \alpha) = \sum_{f \in \alpha} \sum_{f'\in \alpha} D(f, f')
\end{equation}
vanishes, then this denotes the probability of the coarse grained history $\alpha$. $\blacksquare$
\end{proof}
Given the decoherence functional in its general form (\ref{eq:generaldecfunc}), the physically interesting question is for which choice of coarse graining it vanishes. This will involve computing sums of truncated transitions over large numbers of two-complices, which is not presently well-understood; instead, current focus is on the computation these amplitudes under certain approximations (see section \ref{qchistories}).
\section{Coarse grainings of spin network histories} \label{coarsegraining}
\subsection{Cosmological coarse graining}
Assertions about certain values for extrinsic and intrinsic curvature clearly are an example of such a coarse graining. This is a natural choice, since intrinsic and extrinsic curvature are the variables of the classical phase space of general relativity. In order to implement this coarse graining on loop gravity, let me stress the following:
\begin{enumerate}
	\item As generic states of a given geometry are given by superpositions of spin network states, one needs to specify what it means to talk about states having particular values for extrinsic and intrinsic geometry. This is accomplished by coherent states as introduced in section I.
	\item The kinematical Hilbert space of LQG is spanned by coherent states of the form (\ref{eq:coherent}). This is due to the overcompleteness property of coherent states. Hence without loss of generality we can use the overcomplete basis of coherent states to give boundary states of fine grained histories.
	\item The set of states which satisfy the peakedness property but not the annihilation operator, overcompleteness and minimal uncertainty properties is generally larger than the set of coherent states. However, following the literature I work with coherent states as these are precursors to semiclassical states.
\end{enumerate} 
Using the machinery of coherent states described above, one can partition the kinematical diff-invariant Hilbert space of loop gravity into classes for which states are peaked around certain values of intrinsic and extrinsic curvature.
\begin{definition} \emph{(Cosmological coarse graining.)}\\
The diffeomorphism-invariant partition of the set of fine grained histories of spin networks in the kinematical Hilbert space $\mathcal{H}_{kin}$ of loop gravity given by the specification of intrinsic and extrinsic curvature of a three-geometry is implemented by considering coherent spin network states peaked on $SL(2,\mathbb{C})$ values $H_l = e^{2itL_l}h_l$. Call this the \emph{cosmological coarse graining}.
\end{definition}

Consider the Hilbert space $\mathcal{H}_{kin}:= \mathcal{H}_\text{f}^*  \otimes \mathcal{H}_\text{i} $, and suppose that $\psi_\text{f}  \otimes \psi_\text{i} , \psi_\text{f} '\otimes \psi_\text{i} ' \in \mathcal{H}_{kin}$. Fine grained histories of these spin network states are specified by spin foam transitions on two-complices $\mathcal{C}, \mathcal{C}'$, respectively. Suppose that we implement the cosmological coarse graining on the space of fine grained histories $\mathfrak{f}$. The diff-invariant properties are encoded in the coherent states $ \psi_{H_l},  \psi_{H'_l}$ coherent states peaked on values $h_l, h'_l \in SU(2)$ and $L_l, L'_l \in \mathfrak{su}(2)$, respectively. 

Explicitly, suppose that we coarse grain by asking for the amplitude of a state to be peaked on the values $(h_l, L_l)$ at least once along its history. This is given by the expression
\begin{equation}
		\langle \psi_\text{f} |C_{\alpha} | \psi_\text{i} \rangle = \sum_{\mathcal{C}_{\alpha},j_f, H_l} \prod_f (2j_f + 1) \prod_v A_v(H_l),
\label{eq:cosmclassop}
\end{equation} 
where the sum stretches over all those bulk configurations such that there is \emph{at least one} vertex amplitude $A_v(H_l)$ in the product for which $H_l = e^{2itL_l}h_l$. Two such histories $\alpha, \alpha'$ for which the coherent state values are $(h_l, L_l), (h'_l, L'_l)$, respectively, decoherence is measured by the function 
\begin{align}
	D(\alpha, \alpha') = \mathcal{N} \sum_{\psi_\text{f}  \otimes \psi_\text{i} } \langle \psi_\text{f}  |C_{\alpha} | \psi_\text{i}  \rangle \langle \psi_\text{i}  |C_{\alpha'} | \psi_\text{f}  \rangle \nonumber \\
	 =  \mathcal{N} \sum_{\psi_f \otimes \psi_i} \sum_{\mathcal{C}_{\alpha},j_f, H_l} \prod_f (2j_f + 1) \prod_v A_v(H_l) \nonumber \\ \times \sum_{\mathcal{C}_{\alpha'},j'_f, H'_l} \prod_f (2j'_f + 1) \prod_v A_v(H'_l)
\label{eq:cohstatedf}
\end{align}
for any $\alpha, \alpha' \in \{c_\alpha \}$ and a suitable normalization $\mathcal{N}$.
It is not immediately obvious whether coarse grained histories of coherent spin networks peaked on specific intrinsic and extrinsic geometries decohere. However, we know independently that coherent states peaked on different such values are kinematically orthogonal (or approximately so). Hence, cosmologically coarse grained histories decohere if the respective values of intrinsic and extrinsic curvature around which they are peaked are sufficiently far apart.

\subsection{Volume coarse graining}
Here, I briefly mention another possibility. A coarse graining of histories by particular values of volume have been considered in \citep{Hartle:1992as, Craig2011}. Explicitly, consider the coarse graining which consists in a partition of the history space into the following classes:\\

$\tilde{c}$: The class of metrics for which all spacelike three-surfaces have volumes less than a fiducial volume $V_0$.\\
$\tilde{c}'$: The class of metrics for each of which there is at least one three-surface with a volume larger than $v_0$.

In loop gravity, spin network states diagonalise the volume operator, the spectrum of which is discrete. Thus, we can straightforwardly translate the above proposal to loop gravity via the following partition of the spin network state space:\\

\begin{definition} \emph{(Volume coarse graining.)}\\
The diffeomorphism-invariant partition of the set of fine grained histories of spin networks in the kinematical Hilbert space $\mathcal{H}_{kin}$ of loop gravity given by the specificiation of a fiducial volume $V_0$ defines a \emph{volume coarse graining} of spin network histories into the following classes:\\
$c$: The class of fine grained spin network histories, all of which consist of spin networks with volume eigenvalues less than a fiducial volume $V_0$;\\
$c'$: The class of fine grained spin network histories in which there is at least one spin network with eigenvalue larger than $V_0$.
\end{definition}

This gives another example of a well-defined coarse graining in loop gravity that has physical meaning. The details of this prescription will depend on the specifics of the spectrum of the volume operator.

\section{Quasiclassical trajectories}\label{qchistories}
Specifying a fine and coarse grained state space as well as a suitable decoherence functional completes the definition of a generalised quantum mechanics, and in particular a coarse grained dynamics. Coarse graining, viz. the partitioning of the history space into diff-invariant classes, amounts to the specification of certain macroproperties for that system. For the theory to have the right coarse grained behaviour, its coarse grained dynamics should thus follow quasi-classical dynamics in the appropriate limit. In the following, I show that there is good evidence for this in case of cosmological coarse graining of spin network histories.

As mentioned above, we presently do not have a good understanding of the infinite sum over bulk configurations needed in calculating the full physical transition amplitude of a boundary spin network. However, it is possible to compute the transition amplitude approximately by implementing what is referred to as the \emph{vertex expansion}. As detailed in \cite[pp.~14]{Oriti:2001qu}, there is another interpretation of the amplitude (\ref{eq:fullphysicalampl})
\begin{equation}
	\langle W | \psi \rangle = \sum_{\mathcal{C}}\sum_{j_f,\text{v}_e} \prod_f (2j_f + 1) \prod_v A_v(j_f, \text{v}_n)
\end{equation}
This can be viewed as an \emph{expansion} of the full physical transition amplitude in orders of vertex amplitudes, similar to the expansion in Feynman graphs. The result is a sum over contributions with increasing numbers of vertex amplitudes. As in the case of standard QED, one takes the leading order as the dominant contributions (for details, see \citep{Rovelli:2011eq, Perini2009}).

In their \cite{Bianchi:2010zs}, Bianchi \emph{et al.} consider the spin foam\footnote{In the following, I only consider the case where the cosmological constant is zero. The non-vanishing case is largely similar; details can be found in \cite{Bianchi:2011ym}.} on a two-complex $\mathcal{C}$ with disconnected boundary of two `dipole graphs' $\Delta_2^*$. These graphs consist of two four-valent nodes $ \{n_1, n_2\} $ and connected by four links $ \{ l_1, l_2, l_3, l_4 \} $. Consider the Hilbert space $\mathcal{H}_{ \Delta_2^*}$ of coherent spin network states on the dipole graphs, labelled by $SL(2,\mathbb{C})$-elements $H_l = e^{i L_l} h_l$, for $h_l \in SU(2)$ and $L_l \in \mathfrak{su}(2)$. There are three steps taken along the way before the classical limit is imposed to obtain quasiclassical trajectories: (A) homogeneity and isotropy, (B) the vertex expansion and (C) the volume expansion.
\subsection{Homogeneity and Isotropy}
First, note that there is an alternative decomposition of $SL(2,\mathbb{C})$-labels according to
\begin{equation}
	H_l = n_{s(l)} e^{-i(\xi_l + i \eta_l) \frac{\sigma_3}{2}} n_{t(l)}^{-1},
\label{eq:sl2decomp2}
\end{equation}
where $\vec{\sigma} = \{ \sigma_i \}, i = 1,2,3$ are the Pauli matrices and $n_l \in SU(2)$. The geometric interpretation of the quadruple $(n_s, n_t, \xi, \eta)$ is explained by Freidel and Speziale in \cite{Freidel:2010aq}: Given appropriate four-valent states with intrinsic and extrinsic curvature, the $n_s, n_t$ are 3d-normals to the triangles of the tetrahedra bounded by the triangle; $\eta$ is the area of the triangle divided by $8 \pi \gamma G \hbar$; and $\xi$ is the sum of the extrinsic curvature at the triangle and the 3d rotation due to the spin connection at the triangle. 

The authors impose homogeneity and isotropy on the boundary states by computing $(h_l, L_l)$ to be 
\begin{equation}
	U_l = n_l e^{i\alpha c \frac{\sigma_3}{2}}n^{-1}_l, \text{ } E_l = -i n_l \frac{2 \pi G \gamma}{t}\beta p \frac{\sigma_3}{2}n^{-1}_l,
\label{eq:homogiso}
\end{equation}
where $n_l$ are $SU(2)$ group elements such that \begin{equation}
n_l \sigma_3 n^{-1}_l = \vec{n}_l \cdot \vec{\sigma}, 
\label{eq:00}
\end{equation}and $\alpha, \beta$ are constants. This entails that $n_{s(l)} = n_{t(l)} = n_l$ and 
\begin{equation}
	\xi_l = \xi = \alpha c, \text{ } \eta_l = \eta = \beta p,
\label{eq:gfhj}
\end{equation}
such that 
\begin{equation}
	H_l (\xi, \eta) = n_l e^{-i (\xi + i\eta) \frac{\sigma_3}{2}} n_l^{-1}
\label{eq:0}
\end{equation} The fact that neither $\xi$ nor $\eta$ depend on $l$ can be seen as the effect of isotropy and the equality (\ref{eq:00}) as the result of homogeneity. Homogenous and isotropic coherent states, then, are labelled by $c = \xi / p$ and $p = \eta / \beta$.

We want to obtain the amplitude for initial and final states labelled by $(\xi_\text{i} , \eta_\text{i} )$ and $(\xi_\text{f} , \eta_\text{f} )$, respectively,
\begin{equation}
	W(\xi_\text{i}, \eta_\text{i} ; \xi_\text{f}, \eta_\text{f}) = W(H_l (\xi_\text{i}, \eta_\text{i}), H_l (\xi_\text{f}, \eta_\text{f}))
\label{eq:1}
\end{equation}
which is holomorphic function of $z_\text{i}, z_\text{f} $ where 
\begin{equation}
	z = \xi + i \eta
\label{eq:2}
\end{equation}
Thus, rewrite the above as
\begin{equation}
	W(z_\text{i}, z_\text{f}) = W(\xi_\text{i}, \eta_\text{i}; \xi_\text{f} , \eta_\text{f} )
\label{eq:3}
\end{equation}
Further modifying the notation, denote
\begin{equation}
	\psi_z(h_l) := \psi_{H_l(z(c,p))}(h_l) := \langle h_l | z \rangle
\label{eq:4}
\end{equation}
using which we obtain the expression
\begin{equation}
	W(z_\text{i}, z_\text{f} ) = \langle \bar{z}_\text{f}  | z_\text{i} \rangle_{physical},
\label{eq:5}
\end{equation}
where, as before, the bracket $\langle \cdot | \cdot \rangle_{physical}$ has the interpretation of the physical inner product.
\subsection{Vertex expansion}
The amplitude (\ref{eq:3}) is computed to \emph{first order} in the vertex expansion; that is, the amplitude is given by a spin foam formed by a single vertex connected to the four boundary nodes by internal lines (see figure \ref{fig:vertexexp}).
\begin{figure}[htb]
\includegraphics[scale=0.8]{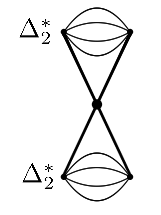}
\caption{Transition amplitude between two `dipole' graphs in the one-vertex expansion. Internal lines are drawn thicker. (Figure taken from \cite{Vidotto:2011qa}).}
\label{fig:vertexexp}
\end{figure}
The amplitude in this approximation for a single vertex $v$ is given by
\begin{equation}
	W(z_\text{i}, z_\text{f}) = W_{v} (H_l(z_\text{i}), H_l(z_\text{f}))
\label{eq:6}
\end{equation}

\subsection{Large volume expansion}
This amplitude can be computed in the limit where the universe is large. This is given by taking large $p$ in equation (\ref{eq:0}). After some details which can be viewed on p. 5 in \cite[p.~03]{Bianchi:2010zs}, this yields
\begin{equation}
	W(z_\text{i}, z_\text{f}) = N^2 z_\text{i} z_\text{f} e^{- \frac{1}{2t\hbar } (z_\text{i}^2 + z_\text{f}^2)}
\label{eq:7}
\end{equation}
wherein $N$ effectively denotes the norm squared of the Livine-Speziale coherent regular tetrahedron. (\ref{eq:7}) is the transition amplitude between two cosmological homogenous isotropic coherent states.
\subsection{Classical limit}%
The transition amplitude lies in the kernel of the quantum operator
\begin{equation}
	\hat{H}:= \lambda \left ( z^2 -t^2 {\hbar}^2 \frac{d^2}{dz^2} - 3 t \hbar \right )^2,
\label{eq:8}
\end{equation}
where $\lambda$ is a constant. This can be rewritten, using the identification
\begin{equation}
	\hat{z} = z, \text{   } \hat{\bar{z}} = t \hbar \frac{d}{dz}
\label{eq:9}
\end{equation}
as the following:
\begin{equation}
	\hat{H} = \lambda (\hat{z}^2 - \hat{\bar{z}}^2 - 2)^2 
\label{eq:10}
\end{equation}
In the classical limit, we replace operators by classical variables and then take the above equation for large p. This yields (after substituting via equation \ref{eq:2})
\begin{eqnarray}
	H &= &-\frac{3}{8\pi G \gamma ^2} \sqrt{p}c^2 = 0\\
	 &= & - \frac{3}{8\pi G} \dot{a}^2a = 0
\label{eq:11}
\end{eqnarray}
where $c=\gamma \dot{a}$ and $p = a^2$. This is the equation for the standard Friedmann cosmology for an empty universe that is either flat or has no volume. This has been generalised to include the cosmological constant (cf. \cite{Bianchi:2011ym}).

The following has happened in this section: First, the degrees of freedom of the theory were truncated to a finite graph (the dipole); secondly, the fine grained state space of spin networks on this graph was coarse grained by considering coherent states peaked on specific values of extrinsic and intrinsic curvature. Thirdly, homogeneity and isotropy were imposed on the state space. Lastly, the one-vertex and large volume expansion were made, before taking the classical limit. 

Hence, for a coarse graining of the fine grained history space into classes corresponding to different values of extrinsic and intrinsic curvature, the coarse grained dynamics are quasiclassical, provided appropriate expansions are made and the appropriate limit is taken. 

\section{Discussion and Outlook}
In this paper, I have proposed a decoherent histories formulation for loop gravity in the spin foam formalism. A cosmological and volume coarse graining were examined, and it was shown that histories coarse grained according to the former follow quasi-classical trajectories given by Friedman cosmology (resp. de-Sitter cosmology for a non-vanishing cosmological constant).

The central point of this paper is that covariant (spin foam) loop gravity provides a natural framework in which to give precise meaning to previously ill-understood formal expressions regarding decoherent histories in quantum general relativity.

In a series of recent papers, Jonathan Halliwell and James Yearsley point out that the standard definition of class operators as a string of projectors is prone to the Quantum Zeno effect and may need to be modified by the introduction of a complex potential \citep{Halliwell2009, Halliwell2009a, Halliwell2010}. Relatedly, in their \cite{Halliwell2012} the authors remark that a straightforward definition of the class operator via path integrals as used in the present paper may be unphysical as it would also be subject to the Quantum Zeno effect. These issues will be addressed in an upcoming paper \cite{Schroeren2013} on Halliwell-style class operators and the quantum Zeno effect in loop gravity.

The results of this article are tentative to the extent that spin foams form part of an actively evolving research programme for which many more or less substantial amendments and revisions are to be expected. Nonetheless, it should the subject of future work to investigate the conceptual implications of these results for foundational problems in quantum theory.

\section*{Acknowledgements}
I am indebted to my supervisor Carlo Rovelli, without whom this work would not have been possible. In addition, I would like to thank James Yearsley, Jonathan Halliwell, Edward Anderson, Petros Wallden, Kinjalk Lochan, Ed Wilson-Ewing, Simone Speziale, Aldo Riello, Wolfgang Wieland, as well as Leonard Cottrell for helpful comments and discussions. I am supported by the German National Academic Foundation (\emph{Studienstiftung des deutschen Volkes}).


\end{document}